\documentclass[copyright,creativecommons]{eptcs}
\providecommand{\event}{ICE'09} 
\providecommand{\volume}{??}
\providecommand{\anno}{2009}
\providecommand{\firstpage}{1}
\providecommand{\eid}{??}

\usepackage{breakurl}

\usepackage[latin1]{inputenc}
\usepackage{amsmath}
\usepackage{amssymb}
\usepackage{amsthm}
\usepackage{enumerate}
\usepackage{xspace}
\usepackage{mathrsfs}
\usepackage{prooftree}
\usepackage[all,2cell]{xy}
\CompileMatrices 
\UseAllTwocells
\SilentMatrices
\usepackage{stmaryrd}
\usepackage{graphicx}
\usepackage{cite}
\usepackage{float}
\usepackage{multicol}
\usepackage{fancybox}

\renewcommand{\vec}[1]{\mathbf{#1}}

\newcommand{\id}{{\mathsf{I}}}
\newcommand{\unit}{{\mathsf{d}}}
\newcommand{\counit}{{\mathsf{e}}}

\newcommand{\tw}{\mathsf{X}}

\newcommand{\nosignal}{\iota}
\newcommand{\len}[1]{{\#{#1}}}

\newcommand{\Left}{\mathsf{L}}
\newcommand{\Right}{\mathsf{R}}

\makeatletter 
\def\moverlay{\mathpalette\mov@rlay} 
\def\mov@rlay#1#2{\leavevmode\vtop{%
\baselineskip\z@skip \lineskiplimit-\maxdimen 
\ialign{\hfil$#1##$\hfil\cr#2\crcr}}} 
\makeatother

\newcommand\twarr[2]{%
\mathrel{\mathop{\moverlay{\scriptstyle\xrightarrow{\,#1\,}\cr{\lower.2em\hbox{$\scriptstyle{}_{#2}$}}}}}}

\newcommand{\dtrans}[2]{\;\mathrel{\twarr{#1}{#2}}\;}
\newcommand{\trdtrans}[2]{\mathrel{(\dtrans{#1}{#2})}}

\newcommand{\ten}{\mathrel{\otimes}} 
\newcommand{\lollipop}{\mathrel{\multimap}}
\newcommand{\W}{\mathscr{W}}

\newcommand{\ie}{\textit{i.e.\ }}
\newcommand{\nb}{\textit{n.b.\ }}
\newcommand{\fig}[1]{{Fig.~{#1}}}
\newcommand{\figref}[1]{\fig{\ref{#1}}}

\newcommand{\var}[1]{{#1}}

\newcommand{\derivationRule}[3]{{\prooftree{\scriptstyle #1}\justifies{\scriptstyle #2}\using\ruleLabel{#3}\endprooftree}}
\newcommand{\reductionRule}[2]{{\prooftree{\scriptstyle #1}\justifies{\scriptstyle #2}\endprooftree}}

\newcommand{\N}{\ensuremath{\mathbb{N}}}

\newcommand{\typeJudgment}[3]{{\scriptstyle #1\;\ent\; {#2} \,\typ\, {#3}}}

\newcommand{\out}[1]{#1!}
\newcommand{\inp}[1]{#1?}

\newcommand{\lowerPic}[3]{\lower#1\hbox{\includegraphics[#2]{#3}}}

\newcommand{\Defeq}
 {\stackrel{\mathrm{def}}{=}}

\newcommand{\bnfEq}{\; ::= \;}
\newcommand{\bnfSep}{\; | \;}

\newcommand{\ent}{\vdash}
\newcommand{\typ}{\mathrel{:}}

\newcommand{\outLabelSub}[2]{\mathop{\mathsf{out_{#2}}}{\ifx\relax#1\relax\!\else#1\fi}}

\newcommand{\inLabel}[1]{\mathop{\mathsf{in}}{\ifx\relax#1\relax\!\else#1\fi}}
\newcommand{\inLabelSub}[2]{\mathop{\mathsf{in_{#2}}}{\ifx\relax#1\relax\!\else#1\fi}}

\newcommand{\openLabel}[1]{\mathop{\mathsf{open}}{\ifx\relax#1\relax\!\else#1\fi}}
\newcommand{\openLabelSub}[2]{\mathop{\mathsf{open_{#2}}}{\ifx\relax#1\relax\!\else#1\fi}}

\newcommand{\ruleLabel}[1]{\textsc{\scriptsize{(#1)}}}

\newcommand{\set}[2]{\{\,#1\;|\;#2\,\}}

\newcommand{\sort}[2]{(#1,\,#2)}

\newcommand{\comp}{\mathrel{;}}
\newcommand{\labelSep}{}

\newcommand{\prefixAtom}[2]{{{\frac{#1}{#2}}}}
\newcommand{\prefix}[3]%
{\mathchoice{\raise.1pc\hbox{$\prefixAtom{#1}{#2}$}\displaystyle#3} 
{\raise.1pc\hbox{$\prefixAtom{#1}{#2}$}\textstyle#3} 
{\prefixAtom{#1}{#2}\scriptstyle#3} 
{\prefixAtom{#1}{#2}\scriptscriptstyle#3}}

\newcommand{\muBind}[2]{{\mu{#1}.\,{#2}}}

\newtheorem{clm}{Claim}

\newtheorem{defn}[clm]{Definition}
\newtheorem{lem}[clm]{Lemma}
\newtheorem{thm}[clm]{Theorem}

\newtheorem{rmk}[clm]{Remark}

\newlength{\mylength}

{\setlength{\fboxsep}{10pt} 
\setlength{\mylength}{\linewidth}%
\addtolength{\mylength}{-2\fboxsep}%
\addtolength{\mylength}{-2\fboxrule}%
\Sbox 
\minipage{\mylength}%
\setlength{\abovedisplayskip}{0pt}%
\setlength{\belowdisplayskip}{0pt}%
}%
{\endminipage\endSbox 
\[\fbox{\TheSbox}\]}

\title{A non-interleaving process calculus for multi-party synchronisation}
\author{Pawe{\l} Soboci{\'n}ski \institute{ECS, University of Southampton}}	
\def\authorrunning{Pawe{\l} Soboci{\'n}ski}
\def\titlerunning{A non-interleaving process calculus for multi-party synchronisation}

\begin{document}
\maketitle
\begin{abstract}
We introduce the wire calculus. Its  
dynamic features are inspired by Milner's 
\textsc{ccs}: a unary prefix operation, binary 
choice and a standard recursion construct.
Instead of an interleaving parallel composition operator 
there are operators for synchronisation along a common boundary ($;$) 
and non-communicating parallel composition ($\ten$)	.
The (operational) semantics is a labelled transition system
obtained with SOS rules. Bisimilarity 
is a congruence with respect to the operators of the language.
Quotienting terms by bisimilarity results in a compact closed category.
\end{abstract}

\section*{Introduction}

Process calculi such as \textsc{csp}, \textsc{sccs}, \textsc{ccs} and their 
descendants attract ongoing theoretical attention because of their claimed status as 
mathematical theories for the study of concurrent systems. Several 
concepts have arisen out of their study, to list just a few:
various kinds of bisimulation (weak, asynchronous, \textit{etc.}), 
various reasoning techniques and even competing schools of thought 
(`interleaving' concurrency vs `true' concurrency). 

The wire calculus does not have a \textsc{ccs}-like interleaving `$\parallel$' operator;
it has operators `$\ten$' and `$\comp$' instead. 
It retains two of Milner's other
influential contributions: (\textit{i}) bisimilarity as extensional 
equational theory for observable behaviour, and
(\textit{ii}) syntax-directed \textsc{sos}. The result is a process calculus which is `truly'
concurrent, has a formal graphical representation and on which `weak' bisimilarity
agrees with `strong' (ordinary) bisimilarity. 

The semantics of a (well-formed) syntactic expression in the calculus is a
labelled transition system (\textsc{lts}) with a chosen state and 
with an associated sort that describes the structure of the labels in the \textsc{lts}.
%
%
%
%
In the graphical representation this is a `black box'
with a left and a right boundary, where each boundary is simply a
number of ports to which wires can be connected. Components are
connected to each other by `wires'---whence the name of the calculus. 
The aforementioned sort of an \textsc{lts} 
prescribes the number of wires that can connect to each of the boundaries.

Notationally, transitions have upper and lower labels that correspond,
respectively, to the left and the right boundaries.
The labels themselves are words over some
alphabet $\Sigma$ of signals and an additional symbol
`$\nosignal$' that denotes the absence of signal.
%

The operators of the calculus fall naturally into `coordination' 
 operators for connecting components
and `dynamic' operators for specifying the behaviour of components.
There are two coordination operators, `$\ten$' and `$\comp$'. 
The dynamic operators consist of a pattern-matching unary prefix, a \textsc{csp}-like
choice operator and a recursion construct. Using the latter, one
can construct various wire constants, such as different kinds of forks.
These constants are related to Bruni, Lanese and Montanari's stateless
connectors~\cite{Bruni2005a} although there they are primitive entities;
here they are expressible using the primitive operations of the calculus.


\smallskip
The wire calculus draws its inspiration form many sources, the two most 
relevant ones come from the model theory of processes: tile logic~\cite{Gadducci2000} and 
the Span(Graph)~\cite{Katis1997a} algebra. Both frameworks are aimed at modeling the operational 
semantics of systems. Tile logic inspired the use of 
doubly-labeled transitions, as well as the use of specific rule formats 
in establishing the congruence property for bisimilarity. Work on Span(Graph)
informs the graphical presentation adopted for the basic operators 
of the wire calculus and inspired the running example.
Similarly, Abramsky's interaction categories~\cite{Abramsky1995} 
are a related paradigm:
in the wire calculus, processes (modulo bisimilarity) are the arrows of a category
and sequential composition is interaction. There is also a tensor product operation,
but the types of the wire calculus have a much simpler structure.

One can argue, however, that the wire calculus is closer in spirit to 
Milner's work on process algebra than the aforementioned frameworks. The 
calculus offers a syntax for system specification. 
The semantics is given by a small number of straightforward SOS rules.
No structural equivalence (or other theory) is required on terms: the
equational laws among them are actually derived via bisimilarity. 
The fact that operations built up from the primitive operators of the
wire calculus preserve bisimilarity means that there is also a shared perspective
with work on \textsc{sos} formats~\cite{Aceto1999} as well as 
recent work on \textsc{sos}-based
coordination languages~\cite{Bliudze2008}. 


\smallskip
An important technical feature is the
\emph{reflexivity} of labelled transition systems---any process can always `do nothing' 
(and by doing so synchronise with $\iota$'s on its boundaries). This idea is also
important in the Span(Graph) algebra; here it ensures that two unconnected components
are not globally synchronised.
The wire calculus also features
$\iota$-\emph{transitivity} which ensures that
unconnected components are not assumed to run at the same speeds and that internal
reduction is unobservable. Reflexivity and $\iota$-transitivity together ensure that
weak and strong bisimilarities coincide---they yield the same equivalence 
relation on processes.

\paragraph{Structure of the paper.}
The first two sections are an exposition of the 
syntax and semantics of the wire calculus. 
In \S\ref{sec:flipFlops} we model a simple synchronous circuit. The next
two sections are devoted to the development of the calculus' theory
and \S\ref{sec:directed} introduces a directed variant.
\section{Preliminaries}

A \emph{wire sort} $\Sigma$ is a (possibly infinite) set of \emph{signals}.
Let $\iota$ be an element regarded as ``absence of signal'', with $\iota\notin\Sigma$.
In order to reduce the burden of bureaucratic overhead, in this
paper we will consider calculi with one sort.
The generalisation to several sorts is a straightforward exercise. 

Before introducing the syntax of the calculus, let us describe those 
labelled transition systems (\textsc{lts}s) that will serve
as the semantic domain of expressions: $\sort{k}{l}$\emph{-components}. In the following
$\len{\vec{a}}$ means the length of a word $\vec{a}$. We shall use $\vec{\iota}$
for words that consist solely of $\iota$'s.
\begin{defn}[Components\label{defn:component}]\rm
Let $L\Defeq \Sigma+\{\iota\}$ and $k,l\geq 0$.
A $\sort{k}{l}$-\emph{transition}
is a labelled transition of the form 
$\dtrans{\vec{a}}{\vec{b}}$ where $\vec{a},\vec{b}\in L^*$, 
$\len{\vec{a}}=k$ and $\len{\vec{b}}=l$.
A $\sort{k}{l}$-\emph{component} $\mathscr{C}$
is a \emph{pointed}, \emph{reflexive} and $\iota$-\emph{transitive} 
\textsc{lts} $(v_0,V,T)$ of $\sort{k}{l}$-transitions. The meanings
of the adjectives are:
\begin{itemize}
\item \emph{Pointed}: there is a \emph{chosen state} $v_0\in V$; 
\item \emph{Reflexive}:
for all $v\in V$ there exists a transition
$v \dtrans{\vec{\iota}}{\vec{\iota}} v$; 
\item $\iota$-\emph{Transitive}: if 
$v \dtrans{\vec{\iota}}{\vec{\iota}} v_1$, $v_1\dtrans{\vec{a}}{\vec{b}} v_2$
and $v_2\dtrans{\vec{\iota}}{\vec{\iota}}v'$ then also
$v \dtrans{\vec{a}}{\vec{b}} v'$.
\end{itemize}
\end{defn}
\begin{tabular*}{\linewidth}[b]{p{.75\linewidth}@{\quad }c}
To the right is a graphical
rendering of $\mathscr{C}$; a box with $k$ wires on the left
and $l$ wires on the right:
&
$\lowerPic{.8pc}{width=2cm}{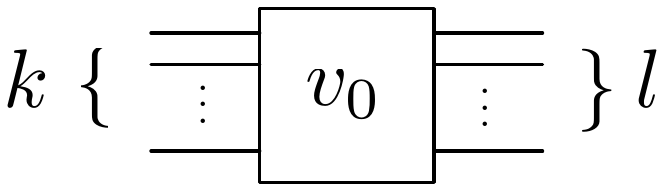}$
\end{tabular*} 

When the underlying \textsc{lts} is clear from the context we
often identify a component with its chosen state,
writing $v_0\typ\sort{k}{l}$ to indicate that $v_0$ is (the chosen state of) a 
$\sort{k}{l}$-component.

\begin{defn}[Bisimilarity]\rm
For two $(k,l)$-components $\mathscr{C},\,\mathscr{C}'$, a simulation
from $\mathscr{C}$ to $\mathscr{C'}$ is a relation $S$ that contains $(v_0,v_0')$
and satisfies the following:  if $(v,w)\in S$ and 
$v\dtrans{\vec{a}}{\vec{b}}v'$ then $\exists w'$ s.t.\ 
$w\dtrans{\vec{a}}{\vec{b}}w'$ and $(v',w')\in S$. 
A bisimulation is a relation $S$ such that itself and its inverse $S^{-1}$ are both simulations.
We say that $\mathscr{C},\,\mathscr{C}'$
are bisimilar and write $\mathscr{C}\sim\mathscr{C}'$
when there is some bisimulation relating $\mathscr{C}$ and $\mathscr{C'}$.
\end{defn}

\begin{rmk}\label{rmk:bisim}\rm
One pleasant feature of reflexive and 
$\iota$-transitive transition systems is that weak bisimilarity
agrees with ordinary (``strong'') bisimilarity, where actions labelled
solely with $\iota$'s are taken as silent. Indeed,
$\iota$-transitivity implies that any weak (bi)simulation in the sense of 
Milner~\cite{Milner1989} is also an ordinary (bi)simulation. 
\end{rmk}

\section{Wire calculus: syntax and semantics}

The syntax of the wire calculus is given in~\eqref{eq:syntax} and 
\eqref{eq:prefix} below:
\begin{equation}\label{eq:syntax}
P  \bnfEq 
\var{Y} \bnfSep P\comp P \bnfSep P\ten P\bnfSep
\prefix{M}{M}{}P \bnfSep P+P \bnfSep \muBind{\var{Y}\typ\tau}{P}
\end{equation}
\begin{equation}\label{eq:prefix}
M \bnfEq \epsilon \bnfSep \var{x} \bnfSep \lambda\var{x} \bnfSep \iota \bnfSep \sigma\in\Sigma 
\bnfSep MM
\end{equation}
All well-formed wire calculus terms have an associated \emph{sort}
$\sort{k}{l}$ where $k,\,l\geq 0$. We let $\tau$ range over sorts.
The semantics of $\tau$-sorted expression is a $\tau$-component.
There is a sort inference system for the syntax presented in \figref{fig:sorting}. In each rule, 
$\Gamma$ denotes the \emph{sorting context}: a finite set 
of (i) process variables with an assigned
sort, and (ii) signal variables.
We will consider only those terms $t$ that have a sort derived from
the empty sorting context: \ie only closed terms.

The syntactic categories in~\eqref{eq:syntax} are, in order:  
process variables,
composition, tensor, prefix, choice and recursion. 
The recursion operator
binds a process variable that is syntactically labelled with a sort. 
When we speak of \emph{terms} we mean abstract syntax 
where identities of bound variables are ignored. There are no additional
structural congruence rules.
In order to reduce the number
of parentheses when writing wire calculus expressions, we shall assume that `$\ten$' binds 
tighter than `$\comp$'.

The prefix operation is specific to the wire calculus and merits an
extended explanation: roughly, it is similar to  
input prefix in value-passing \textsc{ccs}, but has additional pattern-matching features.
A general prefix is of the form $\prefix{u}{v}{P}$ where $u$ and $v$ are strings generated 
by~\eqref{eq:prefix}: $\epsilon$ is the empty string, $x$ is a free occurrence of a 
\emph{signal variable}
from a denumerable set $\vec{\var{x}}=\{x,y,z,\dots\}$ of signal variables (that are disjoint from 
process variables), $\lambda\var{x}$ is the binding of a signal variable, $\iota$
is `no signal' and $\sigma\in\Sigma$ is a signal constant. Note that for a single variable
$x$, `$\lambda x$' can appear several times within a single prefix.
Prefix is a binding operation and binds possibly 
several signal variables that appear freely in $P$; a variable $x$ is bound
in $P$ precisely when $\lambda x$ occurs in $u$ or $v$. Let $bd(\prefix{u}{v}{})$
denote the set of variables that are bound by the prefix, \ie 
$\var{x}\in bd(\prefix{u}{v}) \Leftrightarrow \lambda x \text{ appears in } uv$. Conversely, let
$fr(\prefix{u}{v})$ denote the set of variables that appear freely in the prefix.

\begin{figure}[t]
\[
\reductionRule{\phantom{\Gamma}}{\typeJudgment{\Gamma,\,X\typ\tau}{X}{\tau}} 
\quad 
\reductionRule { \typeJudgment{\Gamma}{P}{\sort{k}{n}} \quad \typeJudgment{\Gamma}{R}{\sort{n}{l}} }
{ \typeJudgment{\Gamma}{P\comp R}{\sort{k}{l}} }
\quad
\reductionRule { \typeJudgment{\Gamma}{P}{\sort{k}{l}} \quad \typeJudgment{\Gamma}{Q}{\sort{m}{n}} }
{ \typeJudgment{\Gamma}{P\ten Q}{\sort{k+m}{l+n}} }
\quad
\reductionRule{ \typeJudgment{\Gamma,\,\var{Y}\typ\tau' }{P}{\tau} }
{ \typeJudgment{\Gamma}{\muBind{\var{Y}\typ\tau'}{P}}{\tau} }
\]
\[
\reductionRule { \len{u}=k,\,\len{v}=l,\,fr(\prefix{u}{v})\cap bd(\prefix{u}{v})=\varnothing,\,fr(\prefix{u}{v}{})\subseteq\Gamma \quad \typeJudgment{\Gamma,\, bd(\prefix{u}{v}{})}{P}{\sort{k}{l}} }
{ \typeJudgment{\Gamma}{ \prefix{u}{v}{P}}{\sort{k}{l}} }
\qquad
\reductionRule{ \typeJudgment{\Gamma}{P}{\tau} \quad \typeJudgment{\Gamma}{Q}{\tau} }
{ \typeJudgment{\Gamma}{P+Q}{\tau} }
\]
\caption{Sorting rules.\label{fig:sorting}}
\end{figure}
When writing wire calculus expressions we shall often omit the 
sort of the process variable when writing recursive definitions.
For each sort $\tau$ there is a term $0_\tau \Defeq \muBind{Y\typ\tau}{Y}$. 
We shall usually omit the subscript;
as will become clear below, $0$ has no non-trivial behaviour. 

The operations `$\comp$' and `$\ten$' have graphical representations that are
convenient for modelling physical systems; this is illustrated below.
\[
\lowerPic{0pc}{width=3cm}{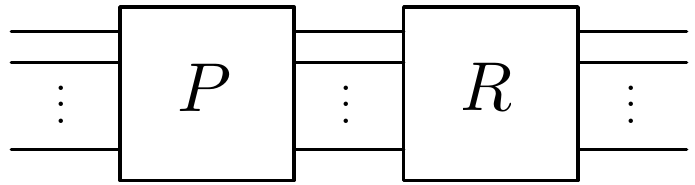}
\qquad \qquad \lowerPic{1.25pc}{width=1.8cm}{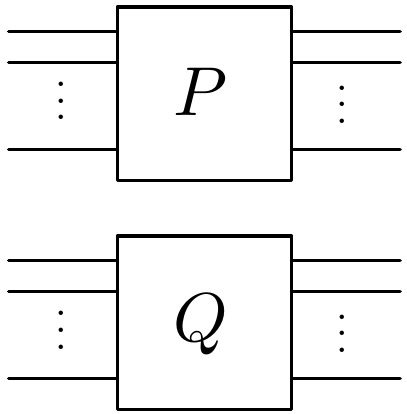}
\]

The semantics of a term $t\typ \sort{k}{l}$ is the  
$\sort{k}{l}$-component 
with \textsc{lts} 
defined using \textsc{sos} rules: 
\ruleLabel{Refl}, \ruleLabel{$\iota$L}, \ruleLabel{$\iota$R},
\ruleLabel{Cut}, \ruleLabel{Ten}, \ruleLabel{Pref}, \ruleLabel{$+\iota$},
\ruleLabel{$+$L}, \ruleLabel{$+$R} and \ruleLabel{Rec} of \figref{fig:sosRules}.
The rules \ruleLabel{Refl}, \ruleLabel{$\iota$L} and \ruleLabel{$\iota$R} guarantee
that the \textsc{lts} satisfies reflexivity and $\iota$-transitivity of 
Definition~\ref{defn:component}. Note that while
\emph{any} calculus with global rules akin to
\ruleLabel{Refl}, \ruleLabel{$\iota$L}, \ruleLabel{$\iota$R} for `silent' actions
is automatically $\iota$-transitive and reflexive with the aforementioned result
on bisimilarity (Remark~\ref{rmk:bisim}),
there are deeper reasons 
for the presence of these rules: see Remark~\ref{rmk:trueConcurrency}.

The rules \ruleLabel{Cut} and \ruleLabel{Ten} justify the intuition
that `$\comp$' is synchronisation along a common boundary and 
`$\ten$' is parallel composition sans synchronisation.
For a prefix $\prefix{u}{v}{P}$ 
a \emph{substitution} is a map $\sigma:bd(\prefix{u}{v}{})\to\Sigma+\{\iota\}$. 
For a term $t$ with possible free occurrences of signal 
variables in $bd(\prefix{u}{v}{})$ 
we write $t|_\sigma$ for the term resulting from the
substitution of $\sigma(x)$ for $x$ for each $x\in bd(\prefix{u}{v})$.
Assuming that $u,\,v$ have no free signal variables, we 
write $u|_\sigma,\,v|_\sigma$ for the strings over $\Sigma+\{\iota\}$ that result from replacing
each occurrence of `$\lambda x$' with $\sigma(x)$. 
The \textsc{sos} rule that governs the behaviour of prefix is \ruleLabel{Pref},
where $\sigma$ is any substitution.
The choice operator $+$ has a similar semantics to that of the \textsc{csp} choice operator 
`$\oblong$', see for instance~\cite{Brookes1988}: a choice is made only on the
occurrence of an observable action. Rules
\ruleLabel{$+\iota$}, \ruleLabel{$+$L} and \ruleLabel{$+$R} make this statement precise.
The recursion construct and its \textsc{sos} rule \ruleLabel{Rec} are standard. 
We assume that the substitution within the \ruleLabel{Rec} rule is non-capturing.
\begin{figure}[t]
\[
\derivationRule{}{P\dtrans{\vec{\iota}}{\vec{\iota}} P}{Refl}
\qquad
\derivationRule{P\dtrans{\vec{\iota}}{\vec{\iota}} R \quad R \dtrans{\vec{a}}{\vec{b}} Q}
{P\dtrans{\vec{a}}{\vec{b}}Q}{$\iota$L}
\qquad
\derivationRule{P\dtrans{\vec{a}}{\vec{b}} R\quad R\dtrans{\vec{\iota}}{\vec{\iota}} Q}
{P\dtrans{\vec{a}}{\vec{b}}Q}{$\iota$R}
\]
\[
\derivationRule{P\dtrans{\vec{a}}{\vec{c}} Q \quad R\dtrans{\vec{c}}{\vec{b}} S}
{P\comp R \dtrans {\vec{a}}{\vec{b}} Q\comp S}{Cut}
\qquad
\derivationRule{P\dtrans{\vec{a}}{\vec{b}} Q\quad R\dtrans{\vec{c}}{\vec{d}} S}
{P\ten R \dtrans{\vec{a}\labelSep \vec{c}}{\vec{b}\labelSep \vec{d}} Q\ten S}{Ten}
\]
\[
\derivationRule{} 
{\prefix{u}{v}{P} \dtrans{u|_\sigma}{v|_\sigma} P|_\sigma}{Pref} 
\quad 
\derivationRule{P[\muBind{\var{Y}}{P}/\var{Y}]\dtrans{\vec{a}}{\vec{b}}Q}
{\muBind{\var{Y}}{P} \dtrans{\vec{a}}{\vec{b}} Q}{Rec}
\quad
\derivationRule{P \dtrans{\vec{\iota}}{\vec{\iota}} Q\ R\dtrans{\vec{\iota}}{\vec{\iota}} S} 
{P + R\dtrans{\vec{\iota}}{\vec{\iota}} Q + S}{$+\iota$}
\quad
\derivationRule{P \dtrans{\vec{a}}{\vec{b}} Q \quad (\vec{ab}\neq\vec{\iota})}
{P+R \dtrans{\vec{a}}{\vec{b}} Q}{$+$L} 
\]
\caption{Semantics of the wire calculus. Symmetric
rule 
\ruleLabel{$+$R} omitted.
\label{fig:sosRules}}
\end{figure}

As first examples of terms we introduce two 
\emph{wire constants}\footnote{Those terms whose \textsc{lts} has a single state.},
their graphical representations
as well as an \textsc{sos} characterisation of their semantics.
\[
\id\Defeq \muBind{Y}{\prefix{\lambda x}{\lambda x}{Y}}:\sort{1}{1} 
\qquad
\lowerPic{-.1pc}{width=.5cm}{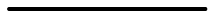} 
\qquad
\derivationRule{}
{\id  \dtrans{a}{a} \id }{Id}
\]
\[
\tw\Defeq \muBind{Y}{\prefix{\lambda x\lambda y}{\lambda y\lambda x}{Y}}:\sort{2}{2} 
\qquad
\lowerPic{.2pc}{width=.5cm}{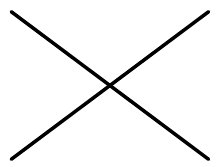}
\qquad 
\derivationRule{}
{\tw \dtrans{a\labelSep b}{b\labelSep a} \tw}{Tw}
\]
\begin{rmk}\rm
Wire constants are an important concept of the wire calculus. Because
they have a single state, any expression built up from
constants using `$\comp$' and `$\ten$' has a single state. Bisimilar
wirings can thus be substituted in terms not only without altering
externally observable behaviour, but also without combinatorially
affecting the (intensional)
internal state.
\end{rmk}


\section{Global synchrony: flip-flops} \label{sec:flipFlops}
As a first application of the wire calculus, we shall model the following  
circuit\footnote{This example was proposed to the author by John Colley.}:
\begin{equation}\label{picture:toggleSystem}
\lowerPic{.8pc}{height=.8cm}{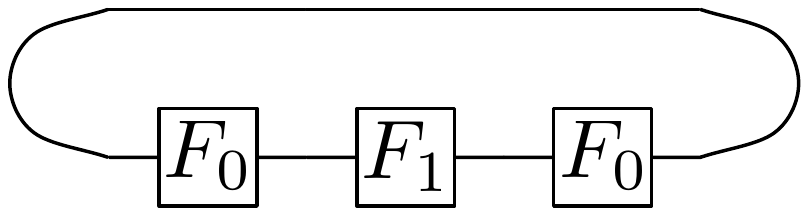}
\end{equation}
where the signals that can be sent along the wire are $\{0,1\}$ (\nb $0$ is
not the absence of signal, that is represented by $\iota$). The
boxes labelled with $F_0$ and $F_1$ are toggle switches
with one bit of memory, the state of which is manifested by the subscript. 
They are simple abstraction of a flip-flop. The switches 
synchronise to the right only on a signal that corresponds to their current state and change 
state according to the signal on the left.
The expected behaviour is a synchronisation
of the whole system --- here a tertiary synchronisation.
In a single ``clock tick'' the middle component will change state to
$0$, the rightmost component to $1$ and the leftmost component will
remain at $0$. 
$F_0$ and $F_1$ are clearly symmetric and their behaviour is
characterised by the \textsc{sos} rules below (where $i\in\{0,1\}$)	.
\[
\derivationRule{}   
{F_i \dtrans{i}{i} F_i}{$i$set$i$}
\quad
\derivationRule{}
{F_i \dtrans{1-i}{i} F_{1-i}}{$i$set$1-i$}
\quad
\derivationRule{}
{F_i \dtrans{\iota}{\iota} F_i}{$i$Refl}
\]
In the wire calculus $F_0$ and $F_1$ can be defined by the terms:
\[
F_0\Defeq \muBind{\var{Y}}{\prefix{0}{0}{\var{Y}} + 
	\prefix{1}{0}{\muBind{\var{Z}}{(\prefix{1}{1}{\var{Z}}+\prefix{0}{1}{\var{Y}})}}}
\text{ and } 
F_1\Defeq \muBind{\var{Z}}{\prefix{1}{1}{\var{Z}} +
	\prefix{0}{1}{\muBind{\var{Y}}{(\prefix{0}{0}{\var{Y}} + \prefix{1}{0}{\var{Z}})}}}.
\]
In order to give an expression for the whole circuit, we need
two additional wire constants $\unit$ and $\counit$.
They are defined below, together with a graphical representation 
and an \textsc{sos} characterisation. Their mathematical significance 
will be explained in Section~\ref{sec:closed}.
\[
{\unit\typ\sort{0}{2}} \Defeq \muBind{Y}{\prefix{}{\lambda x \lambda x}{Y}} 
\qquad
\lowerPic{.4pc}{width=.5cm}{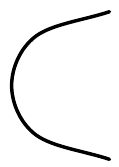} 
\qquad
\derivationRule{}
{\unit \dtrans{}{a\labelSep a} \unit}{$\unit$} 
\]
\[
{\counit\typ\sort{2}{0}} \Defeq \muBind{Y}{\prefix{\lambda x \lambda x}{}{Y}} 
\qquad
\lowerPic{.4pc}{width=.5cm}{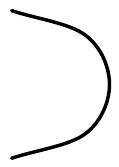} 
\qquad
\derivationRule{}
{\counit \dtrans{a\labelSep a}{} \counit}{$\counit$}
\]
Returning to the example,
the wire calculus expression corresponding to~\eqref{picture:toggleSystem}
can be written down in the wire calculus
by scanning the picture from left to right.
\[A \Defeq \unit \comp (\id \ten (F_0\comp F_1\comp F_0)) \comp \counit.\]
What are its dynamics?  
Clearly $A$ has the sort $\sort{0}{0}$. It is immediate that any two terms
of this sort are extensionally equal (bisimilar) as
there is no externally observable behaviour. This should be intuitively
obvious because $A$ and other terms of sort $\sort{0}{0}$ are closed systems; 
they have no boundary on which an observer can interact. We can,
however, examine the intensional internal state of the system and the possible internal state
transitions.
The semantics is given 
structurally in a compositional way -- so any non-trivial behaviour of $A$
is the result of non-trivial behaviour of its components.
Using \ruleLabel{Cut}
the only possible behaviour of $F_0\comp F_1$ is of the form (i) below.
\[
\derivationRule{
F_0\dtrans{x}{z} X \quad
F_1\dtrans{z}{y} Y
}
{ F_0\comp F_1 \dtrans{x}{y} X\comp Y }{Cut}
\quad
\derivationRule{
F_0\dtrans{x}{0} X \quad
F_1\dtrans{0}{1} F_0
}
{ F_0\comp F_1 \dtrans{x}{1} X\comp F_0 }{Cut}
\quad 
\derivationRule{
F_0 \comp F_1 \dtrans{x}{1} X\comp F_0  \quad
F_0\dtrans{1}{0} F_1
}
{ F_0\comp F_1 \comp F_0 \dtrans{x}{0} X\comp F_0 \comp F_1 }{Cut}
\]
\[
\scriptstyle (i) \quad\qquad\qquad\qquad\qquad\qquad (ii) \qquad\quad\qquad\qquad\qquad\qquad (iii)
\]
The dynamics of $F_0$ and $F_1$ imply that $z=0$ and $y=1$, and the latter
also implies that $Y=F_0$, hence the occurrence of \ruleLabel{Cut} must be of the
form (ii) above. Composing with $F_0$ on the right yields (iii).
\[
\derivationRule{}
{
\id\ten (F_0\comp F_1 \comp F_0) \dtrans{w\labelSep x}{w\labelSep 0} \id\ten (X\comp F_0 \comp F_1)
}
{$\ten$}
\]
\[
\scriptstyle (iv)
\]
Next, tensoring with $\id$ yields (iv) above.
In order for \ruleLabel{Cut} to apply after post-composing with $\counit$, 
$w$ must be equal to $0$, yielding (iv) below.
\[
\derivationRule{}
{
(\id\ten (F_0 \comp F_1 \comp F_0)) \comp \counit 
\dtrans {0\labelSep x}{} 
(\id\ten (X\comp F_0 \comp F_1)) \comp \counit
}{Cut} \qquad
\derivationRule{}
{
\unit\comp (\id\ten (F_0 \comp F_1\comp F_0)) \comp \counit \dtrans {}{} 
\unit\comp (\id\ten (F_0\comp F_0 \comp F_1)) \comp \counit
}{Cut}
\]
\[
\scriptstyle (iv) \quad\qquad\qquad\qquad\qquad\qquad\quad(v)
\]
The final application of \ruleLabel{Cut} forces $x=0$ and hence $X=F_0$, 
resulting in (v). It is easy to repeat this argument to show that the
entire system has three states --- intuitively  `$1$ goes around the ring'.
It is equally easy to expand this example to synchronisations that involve 
arbitrarily many components. Such behaviour is not easily handled using
existing process calculi.

\begin{rmk}\label{rmk:trueConcurrency}\rm
Consider the wire calculus expression represented by 
 ($\dag$).
\[
\lowerPic{1.7pc}{height=1.5cm}{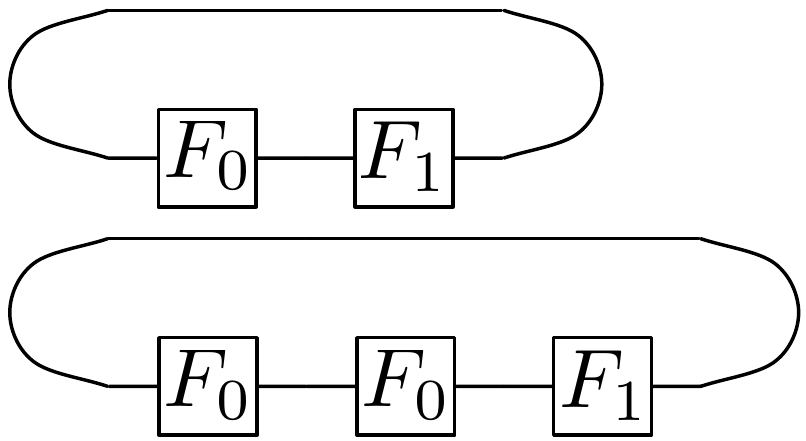}\quad (\dag)
\]
The two components are unconnected. 
As a consequence of \ruleLabel{Refl} and \ruleLabel{Ten},
a ``step'' of the entire system is either a (i) step of the
upper component, (ii) step of the lower component or (iii) a `truly-concurrent' step
of the two components. The presence of the rules
\ruleLabel{$\iota$L} and \ruleLabel{$\iota$R} then 
ensures that behaviour is not scheduled
by a single global clock --- that would be unreasonable for many important scenarios.
Indeed, in the wire calculus synchronisations occur only between explicitly 
connected components.
\end{rmk}


\section{Properties}
\label{sec:properties}
 
In this section we state and prove the main properties of the wire calculus: bisimilarity
is a congruence and terms up-to-bisimilarity are the arrows of a symmetric monoidal category
$\W$. In fact, this category has further structure that will be explained in 
section~\ref{sec:closed}.	
 
Let $P\trdtrans{\vec{a}_k}{\vec{b}_k} Q$ denote trace
$P\dtrans{\vec{a}_1}{\vec{b}_1}P_1\cdots P_{k-1} \dtrans{\vec{a}_k}{\vec{b}_k} Q$,
for some $P_1,\dots, P_{k-1}$.
\begin{lem}\label{lem:traces}Let $P\typ\sort{k}{l}$ and $Q\typ\sort{l}{m}$, then:
\begin{enumerate}[(i)]
\item
If $P\comp Q\dtrans{\vec{a}}{\vec{b}} R$ then $R=P'\comp Q'$ and there exist traces
\[
P\trdtrans{\vec{\iota}}{\vec{d}_k} 
P_{k} \dtrans{\vec{a}}{\vec{c}} 
P_l' 
\trdtrans{\vec{\iota}}{\vec{e}_{l}} P'
,\quad 
Q\trdtrans{\vec{d}_k}{\vec{\iota}} 
Q_k \dtrans{\vec{c}}{\vec{b}} Q_l' 
\trdtrans{\vec{e}_l}{\vec{\iota}} Q';
\]
\item 
If $P\ten Q\dtrans{\vec{a}}{\vec{b}} R$ then $R=P'\ten Q'$ and there exist transitions:
\[
P \dtrans{\vec{a}_1}{\vec{b}_1} P',\quad 
Q \dtrans{\vec{a}_2}{\vec{b}_2} Q' \text{ with }
\vec{a}=\vec{a}_1\vec{a}_2\text{ and }\vec{b}=\vec{b}_1\vec{b}_2.
\]
\end{enumerate}
\end{lem}
\begin{proof}
Induction on the derivation of $P\comp Q\dtrans{\vec{a}}{\vec{b}} R$. The base cases are:
a single application of \ruleLabel{Cut} whence $P\dtrans{\vec{a}}{\vec{c}} P'$ and 
$Q\dtrans{\vec{c}}{\vec{b}} Q'$
and $R=P'\comp Q'$, and a single application of \ruleLabel{Refl} whence $R=P\comp Q$ but
by reflexivity we have also that $P\dtrans{\vec{\iota}}{\vec{\iota}} P$ and
$Q\dtrans{\vec{\iota}}{\vec{\iota}} Q$. The first inductive step is an application of
\ruleLabel{$\iota$L} where $P\comp Q\dtrans{\vec{\iota}}{\vec{\iota}} R$ and 
$R\dtrans{\vec{a}}{\vec{b}} S$. Applying the inductive hypothesis to 
the first transition we get $R=P''\comp Q''$ and suitable traces. Now we apply the
inductive hypothesis again to $R''\comp Q''\dtrans{\vec{a}}{\vec{b}} S$ and
obtain $S=P'\comp Q'$ together with suitable traces. The required traces are obtained
by composition in the obvious way. 
The second inductive step is an application of 
\ruleLabel{$\iota$R} and is symmetric.
Part (ii) involves a similar, if easier, induction.
\end{proof}

\begin{thm}[Bisimilarity is a congruence]\label{thm:coordinationCongruence}
Let $P,Q\typ \sort{k}{l}$, 
$R\typ \sort{m}{n}$, 
$S\typ \sort{l}{l'}$ 
and $T\typ\sort{k'}{k}$ and suppose that $P\sim Q$. Then: 
\begin{enumerate}[(i)]
\item ${P\comp S}\sim{Q\comp S}$ and ${T\comp P}\sim{T\comp Q}$;
\item ${P\!\ten\! R}\!\sim\!{Q\!\ten\! R}$ and ${R\!\ten\! P}\!\sim\!{R\!\ten\! Q}$;
\item $\prefix{u}{v}{P} \sim \prefix{u}{v}{Q}$;
\item ${P+ R} \sim {Q+R}$ and ${R+P} \sim {R+Q}$.
\end{enumerate}
\end{thm}
\begin{proof}
Part (i) follows easily from the conclusion of Lemma~\ref{lem:traces}, (i);
in the first case one shows that 
$A\Defeq \set{(P\comp R,\, Q\comp R)}{P,Q\typ\sort{k}{l},\,R\typ\sort{l}{m},\,P\sim Q}$
is a bisimulation. Indeed, if $P\comp R\dtrans{\vec{a}}{\vec{b}} P'\comp R'$
then there exists an appropriate trace from $P$ to $P'$ and from $R$ to $R'$.
As $P\sim Q$, there exists $Q'$ with $P'\sim Q'$
and an equal trace from $Q$ to $Q'$. We use the traces from $Q$ and $R$ to
obtain a transition
$Q\comp R\dtrans{\vec{a}}{\vec{b}} Q'\comp R'$ 
via repeated applications of \ruleLabel{Cut} and \ruleLabel{$\iota$L,R}. 
Clearly $(P'\comp R',Q'\comp R')\in A$.
Similarly, (ii) follows from Lemma~\ref{lem:traces}, (ii); (iii) and (iv) are
straightforward.
\end{proof}


Let 
$\id_k \Defeq \muBind{Y}{\prefix{\lambda x_1 \dots \lambda x_k}{\lambda x_1\dots \lambda x_k}{Y}}$
and
$\tw_{k,\,l} \Defeq \muBind{Y}{\prefix{\lambda x_1\dots\lambda x_k\lambda y_1\dots\lambda y_l}%
{\lambda y_1\dots\lambda y_l\lambda x_1\dots\lambda x_k}{Y}}$.
$\tw_{k,\,l}$ is drawn as 
$\lowerPic{.7pc}{height=.8cm}{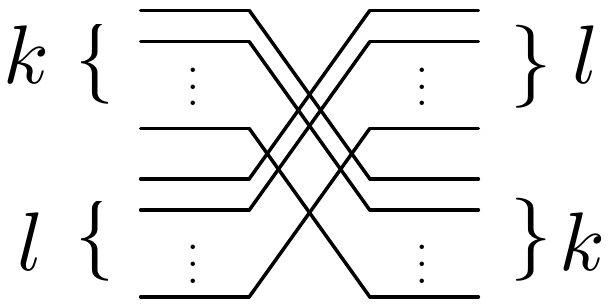}$.

\begin{lem}[Categorical axioms]\label{lem:catAxioms}
 In the 
statements below we implicitly universally quantify over all 
$k,l,m,n,u,v\in \N$, $P\typ\sort{k}{l}$, $Q\typ\sort{l}{m}$, $R\typ\sort{m}{n}$,
$S\typ\sort{n}{u}$ and $T\typ\sort{u}{v}$.
\begin{enumerate}[(i)]
\item ${(P\comp Q)\comp R} \sim {P\comp (Q\comp R)}$;
\item ${P \comp \id_l}  \sim {{P} \ \sim\  {\id_k \comp P}}$;
\item ${(P\ten R)\ten T} \sim {P\ten (R\ten T)}$;
\item ${(P\ten S)\comp (Q\ten T)} \sim {(P\comp Q)\ten (S\comp T)}$;
\item $(P\ten R)\comp \tw_{l,n}\ \sim\ \tw_{k,m}\comp(R\ten P)$; 
\qquad\qquad $\left(\lowerPic{.5pc}{width=1.2cm}{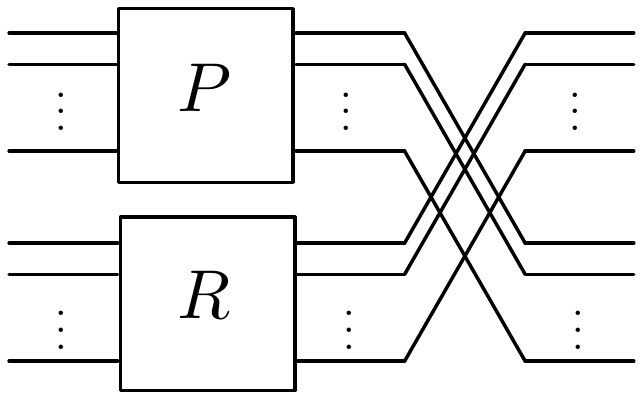}\ \sim\ 
\lowerPic{.5pc}{width=1.2cm}{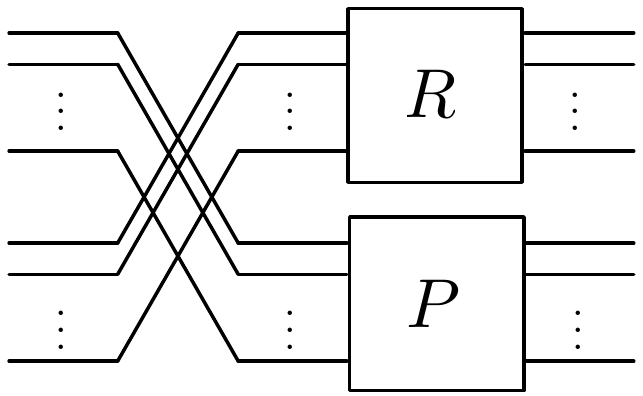}\right)$	
\vspace{-.5pc}
\item  $\tw_{k,l}\comp\tw_{l,k}\ \sim\ \id_{k+l}$.
\end{enumerate}
\end{lem}
\begin{proof}
(i) Here we use Lemma~\ref{lem:traces}, (ii) to decompose a transition
from $(P\comp Q)\comp R$ into traces of $P$, $Q$ and $R$ and then, 
using reflexivity, compose into a transition of $P\comp (Q\comp R)$. 
Indeed, suppose that $(P\comp Q)\comp R \dtrans{\vec{a}}{\vec{b}} (P' \comp Q')\comp R'$.
Then: 
\begin{multline*}
P\comp Q \trdtrans{\vec{\iota}}{\vec{d_i}} P_k\comp Q_k \dtrans{\vec{a}}{\vec{c}}
P_l'\comp Q_l' \trdtrans{\vec{\iota}}{\vec{e_i}} P'\comp Q' 
\text{ and }
R \trdtrans{\vec{d_i}}{\vec{\iota}} 
R_k \dtrans{\vec{c}}{\vec{b}} R_l' 
\trdtrans{\vec{e_i}}{\vec{\iota}} R'.
\end{multline*}
Decomposing the first trace into traces of $P$ and $Q$ yields
\begin{multline*}
P \trdtrans{\vec{\iota}}{\vec{d_{1i}}} P_{1k_1} 
\dtrans{\vec{\iota}}{\vec{c_1}} 
P_{1l_1}' 
\trdtrans{\vec{\iota}}{\vec{e_{1i}}} P_{1} \cdots 
P_k  \trdtrans{\vec{\iota}}{\vec{d_i'}} \dtrans{\vec{a}}{\vec{c}'} 
\trdtrans{\vec{\iota}}{\vec{e_i'}} P_l' \cdots P' \\
\text{ and }
Q \trdtrans{\!\!\vec{d_{1i}}\!\!}{\vec{\iota}} Q_{1k_1} 
\dtrans{\vec{c_1}}{\vec{d_1}} 
Q_{1l_1}'
\trdtrans{\!\!\vec{e_{1i}}\!\!}{\vec{\iota}}
Q_{1}\cdots Q_k
 \trdtrans{\!\!\vec{d_i}'\!\!}{\vec{\iota}} \dtrans{\vec{c}'}{\vec{c}} 
 \trdtrans{\!\!\vec{e_i}'\!\!}{\vec{\iota}} Q_l' \cdots Q'.
\end{multline*}
Using reflexivity and \ruleLabel{Cut} we obtain
\begin{multline*}
Q\comp R  
\trdtrans{\vec{d_{1i}}}{\vec{\iota}} Q_{1k_1}\comp R
\dtrans{\vec{c_1}}{\iota} Q'_{1l_1}\comp R_1 
\trdtrans{\vec{e_{1i}}}{\vec{\iota}} Q_1\comp R_1 \cdots 
Q_k \comp R_k 
\trdtrans{\vec{d_i'}}{\vec{\iota}} \dtrans{\vec{c}'}{\vec{b}} 
\trdtrans{\vec{e_i'}}{\vec{\iota}} Q_l'\comp R_l' \cdots Q'\comp R'.
\end{multline*}
Now by repeated applications of \ruleLabel{Cut} followed by \ruleLabel{$\iota$L,R}
we obtain the required transition $P\comp (Q\comp R) \dtrans{\vec{a}}{\vec{b}} P'\comp(Q'\comp R')$.
Similarly, starting with a transition from $P\comp(Q\comp R)$ one reconstructs a matching
transition from $(P\comp Q)\comp R$.

Parts (ii) and (iii) are straightforward. For (iv), a transition 
$(P\ten R)\comp (Q\ten S) \dtrans{\vec{aa'}}{\vec{bb'}} (P'\ten R')\comp (Q'\ten S')$
decomposes first into traces:
\[
P\ten R \trdtrans{\vec{\iota}}{\vec{d_id_i'}} 
P_k\ten R_k  \dtrans{\vec{aa'}}{\vec{cc'}}
P_l'\ten R_l' \trdtrans{\vec{\iota}}{\vec{e_ie_i'}} P'\ten R'
\]
\[
Q\ten S \trdtrans{\vec{d_id_i'}}{\vec{\iota}} 
Q_k\ten S_k \dtrans{\vec{cc'}}{\vec{bb'}} 
Q_l'\ten S_l' 
\trdtrans{\vec{e_ie_i'}}{\vec{\iota}}
Q'\ten S'
\]
and then into individual traces of $P$, $R$, $Q$ and $S$:
\[
P \trdtrans{\vec{\iota}}{\vec{d_i}}
P_k \dtrans{\vec{a}}{\vec{c}}
P_l' \trdtrans{\vec{\iota}}{\vec{e_i}} P'
\qquad 
R \trdtrans{\vec{\iota}}{\vec{d_i'}} 
R_k  \dtrans{\vec{a'}}{\vec{c'}}
R_l' \trdtrans{\vec{\iota}}{\vec{e_i'}} R'
\]
\[
Q \trdtrans{\vec{d_i}}{\vec{\iota}} 
Q_k \dtrans{\vec{c}}{\vec{b}} 
Q_l' 
\trdtrans{\vec{e_i}}{\vec{\iota}}
Q'
\qquad
S \trdtrans{\vec{d_i'}}{\vec{\iota}}
S_k \dtrans{\vec{c'}}{\vec{b'}} S_l' 
\trdtrans{\vec{e_i'}}{\vec{\iota}}
S'
\]
and hence transitions $P\comp Q \dtrans{\vec{a}}{\vec{b}} P'\comp Q'$ and
$R\comp S \dtrans{\vec{a'}}{\vec{b'}} R'\comp S'$, which combine via a single
application of $\ruleLabel{$\ten$}$ to give 
$(P\comp Q)\ten (R\comp S) \dtrans{\vec{aa'}}{\vec{bb'}} (P'\comp Q')\ten (R'\comp S')$.
The converse is similar. Parts (v) and (vi) are trivial. 
\end{proof}

As a consequence of (i) and (ii) there is a category $\W$ that has the natural numbers
as objects and terms of sort $\sort{k}{l}$ quotiented by bisimilarity as arrows from $k$ to $l$;
for $P\typ\sort{k}{l}$ let $[P]$ to denote $P$'s equivalence class wrt bisimilarity -- 
then $[P]:m\to n$ is the arrow of $\W$. 
The identity morphism on $m$ is  $[\id_m]$. Composition
$[P]:k\to l$ with $[Q]:l\to m$ is $[P\comp Q]:k\to m$.\footnote{Well-defined due to Theorem~\ref{thm:coordinationCongruence}.}
Parts (iii) and (iv) imply that $\W$ is monoidal with a strictly associative tensor.
Indeed, on objects let $m\ten n\Defeq m+n$ and
on arrows $[P]\ten[Q]\Defeq [P\ten Q]$.\footnote{Well-defined due to Theorem~\ref{thm:coordinationCongruence}.}
The identity for $\ten$ is clearly $0$.
Subsequently
(v) and (vi) imply that $\W$ is symmetric monoidal.
Equations (i)--(iv) also justify the use of the graphical syntax: roughly,
rearrangements of the representation in space result in syntactically different but
bisimilar systems.


\section{Closed structure}
\label{sec:closed}

Here we elucidate the closed structure of $\W$: the wire constants $\unit$ and
$\counit$ play an important role.
\begin{lem} \label{lem:unitcounit} 
$\unit\ten\id\comp\id\ten\counit \sim \id \sim \id\ten\unit\comp\counit\ten\id$ \qquad
$\left( \lowerPic{.3pc}{width=.6cm}{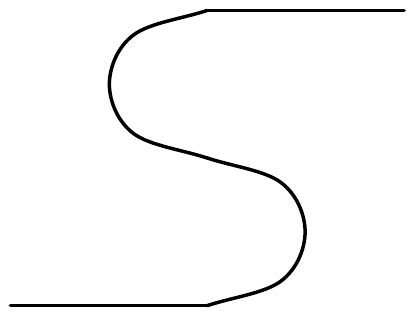} \sim\,\ \lowerPic{-.1pc}{width=.5cm}{identityWire}
\,\ \sim\lowerPic{.3pc}{width=.6cm}{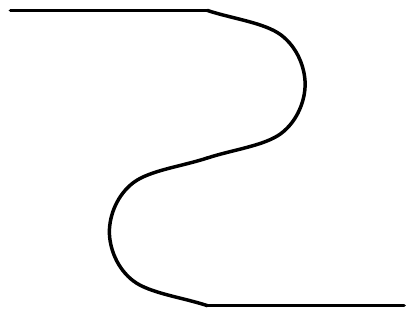} \right)$. 
\end{lem}
Given $k\in \N$, define recursively
\begin{multline*}
\unit_1 \Defeq \unit \text{ and }
\unit_{n+1} \Defeq \unit \comp {\id\ten\unit_{n}\ten\id} 
\text{ and dually }
\counit_1 \Defeq \counit \text{ and } 
\counit_{n+1} \Defeq {I_n\ten\counit\ten I_n} \comp \counit_n.
\end{multline*}
Let $\unit_0,\counit_0=\id_0\Defeq 0$.

\begin{lem}\label{lem:unitcounitgen}
For all $n\in \N$, 
${\unit_n \ten \id_n} \comp {\id_n\ten \counit_n} 
\sim \id_n \sim 
{\id_n\ten \unit_n} \comp {\counit_n\ten\id_n}$. 
\end{lem}
\begin{proof}
For $n=0$ the result is trivial. For $n=1$, it is the conclusion of
Lemma~\ref{lem:unitcounit}. For $n>1$ the result follows by
straightforward inductions: 
\begin{align*}
{\unit_{n+1} \ten \id_{n+1}} \comp\; &{\id_{n+1}\ten \counit_{n+1}} 
\\
&= {(\unit \comp \id\ten\unit_{n}\ten\id)\ten \id_{n+1}} 
   \comp 
   {\id_{n+1}\ten (\id_n\ten\counit\ten \id_n\comp\counit_n )} \\
&\sim {\unit \ten \id_{n+1}} \comp  {\id\ten \unit_n \ten \id_{n+2}}
\comp {\id_{2n+1}\ten \counit \ten \id_n} \comp {\id_{n+1}\ten \counit_{n}} \\
&\sim  {\unit \ten \id_{n+1}} \comp {\id\ten \counit \ten \id_n}
\comp {\id\ten \unit_n \ten \id_n} \comp {\id_{n+1}\ten \counit_n} \\
&\sim  {({\unit \ten \id}\comp {\id\ten \counit})\ten\id_n} \comp
{\id \ten ( \unit_n \ten \id_n \comp \id_n \ten \counit_n)} \\
&\sim I_{n+1}
\end{align*}
\begin{align*}
{\id_{n+1}\ten \unit_{n+1}} \comp\;& {\counit_{n+1}\ten\id_{n+1}} \\
&= { \id_{n+1} \ten ( \unit \comp {\id\ten\unit_{n}\ten\id} )} 
\comp
{ ( {\id_n\ten\counit\ten \id_n} \comp \counit_n ) \ten \id_{n+1} } \\
&\sim {\id_{n+1}\ten \unit} \comp {\id_{n+2}\ten \unit_n\ten \id} \comp
{\id_n \ten \counit \ten \id_{2n+1}} \comp {\counit_n \ten \id_{n+1} } \\
&\sim {\id_{n+1}\ten \unit} \comp {\id_{n} \ten \counit \ten \id}
\comp {\id_n \ten \unit_n \ten \id} \comp {\counit_n \ten \id_{n+1}} \\
&\sim {\id_n \ten ({\id \ten \unit} \comp {\counit \ten \id})} \comp
      {( {\id_n\ten \unit_n} \comp {\counit_n \ten \id_n} ) \ten \id} \\
&\sim \id_{n+1} 
\end{align*}
\end{proof}

With Lemma~\ref{lem:unitcounitgen} we have shown that $\W$ is a compact-closed
category~\cite{Kelly1980}. Indeed, for all $n\geq 0$ let $n^\star\Defeq n$;
we have shown the existence of arrows
\begin{equation}\label{eq:compclosed}
[\unit_n]: 0 \to n\ten n^\star \qquad [\counit_n]:n^\star\ten n\to 0
\end{equation}
that satisfy 
$[\unit_n]\ten id_n \comp \id_n \ten [\counit_n]= \id_n 
= id_n\ten [\unit_n] \comp [\counit_n]\ten \id_n$.\footnote{The $(-)^\star$ operation on 
objects of $\W$
may seem redundant at this point but it plays a role when considering the more
elaborate sorts of directed wires in \S\ref{sec:directed}.}

\smallskip
While the following is standard category theory, it may be useful for the casual
reader to see how this data implies a closed structure.
For $\W$ to be closed wrt $\ten$ we need, for any $l,m\in\W$ an object 
$l\lollipop m$ and a term 
\[ev_{l,m}\typ\sort{(l\lollipop m) \ten l}{m}\]
that together satisfy the following universal property: for any term 
$P\typ \sort{k\ten l}{m}$ there exists 
a unique (up to bisimilarity) term  $Cur(P)\typ\sort{k}{l\lollipop m}$ such that the diagram below 
commutes:
\[
\xymatrix{
{ (l\lollipop m) \ten l } \ar[r]^-{[ev_{l,m}]} & {m} \\
  {k \ten l } \ar[ur]_{[P]} \ar[u]^{[Cur(P)\ten \id_l]}
}
\]
Owing to the existence of $\unit_n$ and $\counit_n$
we have $l\lollipop m \Defeq m+l$, 
$ev_{l,m} \Defeq \id_m \ten \counit_l$ and
$Cur(P)\Defeq {\id_k\ten \unit_l} \comp {P \ten \id_l}$.
Indeed, it is easy to verify universality:
suppose that there exists $Q\typ \sort{k}{m+l}$ such that
${Q \ten \id_l} \comp ev_{l,m} \sim P$. Then 
\begin{align*}
Cur(P) &= {\id_k \ten \unit_l} \comp { P \ten \id_l}  \\
& \sim \id_k\ten \unit_l \comp (Q \ten \id_l \comp ev_{l,m})\ten \id_l \\
& \sim \id_k \comp Q \comp {\id_{m}\ten ( {\id_l \ten \unit_l} \comp {\counit_l \ten \id_l})} \\
& \sim Q
\end{align*}

Starting from the situation described
in \eqref{eq:compclosed} and using completely general reasoning,
one can define a contravariant functor 
$(-)^\star:\W^{op}\to\W$ by $A\mapsto A^\star$ on objects and mapping
$f:A\to B$ to the composite
\[
B^\star \xrightarrow{B^\star \ten \unit_A} B^\star\ten A\ten A^\star 
\xrightarrow{B^\star\ten f\ten A^\star} B^\star\ten B 
\ten A^\star \xrightarrow {\counit_B\ten A^\star}
A^\star.
\]
Indeed, the fact that $(-)^\star$ preserves identities is due to one of the
triangle equations and the fact that it preserves composition is implied by the other.
We can describe $(-)^\star$ on $\W$ directly as the 
up-to-bisimulation correspondent of a structurally recursively defined syntactic transformation
that, intuitively, rotates a wire calculus term by 180 degrees.
First, define the endofunction $(-)^\star$ on  strings generated by~\eqref{eq:prefix} to
be simple string reversal (letters are either free variables, bound variables,
signals or `$\iota$').
\begin{defn}
Define an endofunction $(-)^\star$ on wire calculus terms by structural 
recursion on syntax as follows:
\[
(R\comp S)^\star \Defeq S^\star \comp R^\star \qquad (R\ten S)^\star \Defeq S^\star \ten R^\star
\]
\[
(\prefix{u}{v}{R})^\star \Defeq \prefix{v^\star}{u^\star}{R^\star} 
\qquad 
(R+S)^\star \Defeq R^\star+ S^\star 
\qquad 
(\muBind{Y}{R})^\star \Defeq \muBind{Y}{R^\star}
\]
\end{defn}

Notice that  $\unit^\star = \counit$ and, by a simple structural induction, $P^{\star\star}=P$.
This operation is compatible with $\unit$ and $\counit$ in the following sense.
\begin{lem}
Suppose that $P\typ\sort{k}{l}$. Then 
${\unit_k\comp P\ten \id_k} \sim {\unit_l\comp \id_l \ten P^\star}$. Dually,
$P\ten \id_l\comp\counit_l \sim \id_k\ten P^\star\comp \counit_k$.
\qed
\end{lem}

The conclusion of the following lemma
implies that for any $k,l\geq 0$ one can define a function 
$(-)^\star:\W(k,l)\to\W(l,k)$ by setting $[P]^\star\Defeq [P^\star]$. 
\begin{lem}\label{lem:star}
If $P\typ \sort{k}{l}$ then $P^\star\typ\sort{l}{k}$. Moreover 
$P\dtrans{\vec{a}}{\vec{b}} R$ iff $P^\star \dtrans{\vec{b}^\star}{\vec{a}^\star} {R}^\star$. 
Consequently $P\sim Q$ iff $P^\star\sim Q^\star$.
\end{lem}
\begin{proof}
Structural induction on $P$. 
\end{proof}
Moreover, by definition we have: 
\[
([P];[Q])^\star=[P\comp Q]^\star=[(P\comp Q)^\star]=[Q^\star\comp P^\star]=[Q^\star]\comp[P^\star]=[Q]^\star\comp [P]^\star
\]
and $[I_k]^\star=[I_k]$, thus confirming the functoriality of
$(-)^\star:\W^{op}\to\W$.


\section{Directed wires}
\label{sec:directed}

In the examples of~\S\ref{sec:flipFlops} there is a suggestive yet informal 
\emph{direction} of signals: the components set their state according to the signal
\emph{coming in} on the left and \emph{output} a signal that corresponds
to their current state on the right. It is useful to make this formal so that
systems of components are not wired in unintended ways. Here we take this leap, which
turns out not to be very difficult. Intuitively, decorating wires with directions
can be thought as a type system in the sense that
the underlying semantics and properties are unchanged, the main difference being that
certain syntactic phrases (ways of wiring components) are no longer allowed.

\begin{defn}[Directed components]\label{defn:directedComponent}\rm
Let $\mathcal{D}\Defeq \{\Left,\Right\}$ and subsequently
$L_d \Defeq (\Sigma+\{\iota\})\times\mathcal{D}$.
Abbreviate $(a,\Left)\in L_d$ by $\overleftarrow{a}$ and $(a,\Right)$ by $\overrightarrow{a}$.
Let $\pi:L_d\to \mathcal{D}$ be the obvious projection.

Let $k,l\in \mathcal{D}^*$.
A $\sort{k}{l}$-\emph{transition}
is a labelled transition of the form 
$\dtrans{\vec{a}}{\vec{b}}$ where 
$\vec{a},\vec{b}\in L_d^*$, 
$\pi^*(\vec{a})=k$ and $\pi^*(\vec{b})=l$.
A $\sort{k}{l}$-\emph{component} $\mathscr{C}$
is a \emph{pointed}, \emph{reflexive} and $\iota$-\emph{transitive} 
\textsc{lts} $(v_0,V,T)$ of $\sort{k}{l}$-transitions. 
\end{defn}

The syntax of the directed version of the wire calculus is obtained by replacing~\eqref{eq:prefix}
with~\eqref{eq:directedPrefix} below: signals are now either inputs ($?$)
or outputs ($!$).
\begin{equation}\label{eq:directedPrefix}
M \bnfEq \epsilon \bnfSep D \bnfSep MM 
\end{equation}
\[
D \bnfEq \inp{A} \bnfSep \out{A}
\]
\[
A \bnfEq \var{x} \bnfSep \lambda\var{x} \bnfSep \iota \bnfSep \sigma\in\Sigma 
\]

Define a map $\overline{(-)}:\mathcal{D}^*\to\mathcal{D}^*$ 
in the obvious way by letting
$\overline{\mathsf{L}} \Defeq \mathsf{R}$ and 
$\overline{\mathsf{R}} \Defeq\mathsf{L}$. 
Now define a function $d$ from terms generated by~\eqref{eq:directedPrefix}   
to $\mathcal{D}^*$ recursively by letting $d(\epsilon)\Defeq\epsilon$, 
$d(uv)\Defeq d(u)d(v)$, $d(\inp{x})=\mathsf{L}$ and $d(\out{x})=\mathsf{R}$.
The terms are sorted with the rules in~\figref{fig:sorting}, after replacing
the rule for prefix with the following:
\[
\reductionRule { \overline{du}=k,\,dv=l,\,
{fr(\prefix{u}{v}{})} \cap {bd(\prefix{u}{v}{})} = \varnothing,\,
{fr(\prefix{u}{v}{})}\subseteq\Gamma \quad 
\typeJudgment{\Gamma,\, bd(\prefix{u}{v}{})}{P}{\sort{k}{l}} }
{ \typeJudgment{\Gamma}{ \prefix{u}{v}{P}}{\sort{k}{l}} }
\]
The semantics is defined with the rules of~\figref{fig:sosRules} where labels
are strings over $L_d$ and \ruleLabel{Pref} is replaced with \ruleLabel{dPref} below.
In the latter rule, given a substitution $\sigma$,
let $e_\sigma$ be the function that
takes a prefix component generated by~\eqref{eq:directedPrefix}
with no occurrences of free signal variables
to $L_d^*$, defined by recursion in the obvious way from its definition on atoms:
$e_\sigma(\lambda x?)=(\sigma x,L)$, $e_\sigma(\lambda x!)=(\sigma x,R)$,
$e_\sigma(a?)=(a,L)$, $e_\sigma(a!)=(a,R)$. Abusing notation, let
$\overline{(-)}$ be the endofunction of $L_d^*$ that switches directions of components. 
Then:
\[
\derivationRule{} 
{\prefix{u}{v}{P} \dtrans{\overline{e_\sigma u}}{e_\sigma v} P|_\sigma}{dPref} 
\]
Below are some examples of wire constants in the directed variant of the calculus:
\[
\id_\Left \Defeq \muBind{Y}{\prefix{\lambda x!}{\lambda x?}{Y}}\typ\sort{\Left}{\Left} 
\qquad
\lowerPic{-.1pc}{width=1cm}{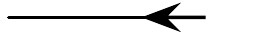} 
\qquad
\derivationRule{}
{\id_\Left  \dtrans{\overleftarrow{a}}{\overleftarrow{a}} \id_\Left }{$\id_\Left$}
\]
\[
\id_\Right \Defeq \muBind{Y}{\prefix{\lambda x?}{\lambda x!}{Y}}\typ\sort{\Right}{\Right}
\qquad
\lowerPic{-.1pc}{width=.8cm}{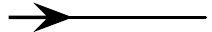} 
\qquad
\derivationRule{}
{\id_\Right  \dtrans{\overrightarrow{a}}{\overrightarrow{a}} \id_\Right}{$\id_\Right$}
\]
\[
\unit_\Left \Defeq \muBind{Y}{\prefix{}{\lambda x?\lambda x!}{Y}}\typ\sort{\epsilon}{\Left\Right}
\quad
\lowerPic{.4pc}{width=.7cm}{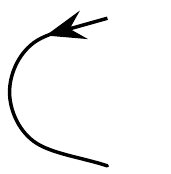} 
\quad
\derivationRule{}
{\unit_\Left \dtrans{\phantom{\quad\ }}{\overleftarrow{a}\labelSep \overrightarrow{a}} \unit_\Left}{$\unit_\Left$}
\]
\[
\counit_\Left \Defeq \muBind{Y}{\prefix{\lambda x?\lambda x!}{}{Y}}\typ\sort{\Right\Left}{\epsilon}
\quad
\lowerPic{.4pc}{width=.5cm}{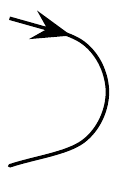} 
\quad
\derivationRule{}
{\counit_\Left \dtrans{\overrightarrow{a}\labelSep \overleftarrow{a}}{} \counit_L}{$\counit_\Left$}
\]
At this point the $F_x$ components of \S\ref{sec:flipFlops} can modelled in the
directed wire calculus. The results of \S\ref{sec:properties} and \S\ref{sec:closed}
hold mutatis mutandis with the same proofs -- this is not surprising as only the
semantics of prefix and the structure of sorts is affected.

\section{Conclusion and future work}

As future work, one should take advantage of the insights of Selinger~\cite{Selinger1997}:
it appears that one can express Selinger's queues and buffers within the wire calculus 
and thus model systems with various kinds of asynchronous communication. 
Selinger's sequential composition, however, has an interleaving semantics.

From a theoretical point of view it should be interesting to test the expressivity of
the wire calculus with respect to well-known interleaving calculi such as CCS 
and non-syntactic formalisms such as Petri nets.

\paragraph{Acknowledgement.} The author thanks the anonymous referees and the participants
in the ICE forum for their suggestions, many of which have improved the paper. Early
discussions with R.F.C. Walters and Julian Rathke were fundamental in stabilising the
basic technical contribution. Special thanks go to Fabio Gadducci and 
Andrea Corradini.
{\small
\bibliographystyle{eptcs}
\bibliography{jab}
}

\end{document}